\documentclass[12pt]{article}
\usepackage{amscd,amsmath,amssymb,amstext,amsthm,exscale,latexsym}
\usepackage{graphicx}
\newcommand {\al}   {\alpha}       \newcommand {\bt}  {\beta}
\newcommand {\g }   {\gamma}       
\newcommand {\dl}   {\delta}       \newcommand {\e }  {\epsilon}

\newcommand {\s }   {\sigma}

         \newcommand {\om}  {\omega}

\newcommand {\pl}   {\partial}     
\renewcommand {\sin}{{\sf\,sin\,}}       \renewcommand {\cos}{{\sf\,cos\,}}

       \renewcommand {\lim}{{\sf\,lim\,}}

\newcommand   {\const}{{\sf\,const}}     \newcommand   {\diag}{{\sf\,diag\,}}

\newcommand {\MO}  {{\mathbb O}}   
   \newcommand {\MR}  {{\mathbb R}}
\newcommand {\MS}  {{\mathbb S}}   
\newcommand {\MU}  {{\mathbb U}}

\newcommand {\CE }  {{\cal E}}

\newcommand {\CM }  {{\cal M}}      
      \newcommand {\CP}  {{\cal P}}

\newcommand {\Sa}  {{\textsc{a}}}   \newcommand {\Sb}  {{\textsc{b}}}

\newcommand {\Sg}  {{\textsc{g}}}   
\newcommand {\Si}  {{\textsc{i}}}   \newcommand {\Sj}  {{\textsc{j}}}
   
   \newcommand {\Sn}  {{\textsc{n}}}

\newcommand {\Ss}  {{\textsc{s}}}

\newtheorem{prop}{Proposition}[section]

\theoremstyle{definition}

\begin{document}
\title     {Nonrelativistic limit of the bosonic string}
\author    {M. O. Katanaev
            \thanks{E-mail: katanaev@mi.ras.ru}\\ \\
            \sl Steklov Mathematical Institute,\\
            \sl 119991, Moscow, ul. Gubkina, 8}
\maketitle
\begin{abstract}
We propose the action for the nonrelativistic string invariant under general
coordinate transformations on the string worldsheet. The Hamiltonian formulation
for the nonrelativistic string is given. Particular solutions of the
Euler--Lagrange equations are found in the time gauge.
\end{abstract}
\section{Introduction}
String theory is one of the basic fields of research in theoretical physics over
last 50 years (see, e.g., \cite{BarNes90,GrScWi87,BriHen88}). String theory is
based on the Nambu--Goto bosonic string theory, whose action was independently
proposed by several authors [4--10].
\nocite{BarChe66A,BarChe66B,Nambu70,Nielse70,Susski70,Hara71,Goto71}
This action is invariant with respect to the Poincar\'e group in the target
space and therefore describes relativistic string.

There is a natural question: ``What happens with the Nambu--Goto action in the
nonrelativistic limit?'' For example, the standard action from classical
Newton mechanics arises in the nonrelativistic limit for a point particle
(see, e.g., \cite{LanLif62}). The problem for the string is contained in the
definition of the nonrelativistic limit. The invariant expansion parameter for
a point particle is given by the ratio of a particle worldline lengthes.
This parameter is not suited for a string because we would like to get an action
which is invariant under general coordinate transformations on the
nonrelativistic string worldsheet. The answer to this question was proposed in
the paper \cite{Katana88}, where the nonrelativistic limit was defined and
exact solution of the equations of motion was found in the form of the
rotating straight rod. Consideration of the nonrelativistic limit for the
bosonic string is important both from theoretical point of view and
applications, for example, in polymer physics.

Note that quantum nonrelativistic string attracted recently much attention as an
individual model \cite{GomOog01,BeGoRoSiYa19}. However, the action and
definition of the limit considered in \cite{Katana88} and \cite{GomOog01}
differ.
\section{The action                                              \label{swbstg}}
The Nambu--Goto Lagrangian is invariant under the Poincar\'e group acting in the
$D$-dimensional Minkowskian space $\MR^{1,D-1}$. Therefore it describes a
relativistic string. The Galilei group is obtained from the Poincar\'e group by
the formal limit of large light velocity $c\to\infty$. Let us point the problem:
we would like to find an action for the bosonic string which is invariant under
general coordinate transformations on the string worldsheet and consistent with
the Galilei transformations in the target space. That is invariant under
translations along all $D$ string coordinates and $\MS\MO(D-1)$-rotations of
the space components. We would like to get this Lagrangian in the
nonrelativistic limit from the Nambu--Goto action. This limit and the model were
proposed in the paper \cite{Katana88}.

The problem is the following. Consider a fixed point on the string worldsheet
$\big(X^\Sa(\tau_0,\s_0)\big)$, where $X^\Sa$, $\Sa=0,1,\dotsc,D-1$, are
Cartesian coordinates in the Minkowski space and $\tau_0,\s_0$ are coordinates
of the point on the string worldsheet. Together with the string, it moves in the
Minkowski space $\MR^{1,D-1}$ along the worldline $\big(X^\Sa(\tau,\s_0)\big)$,
$\tau\in\MR$, where $\s_0$ is a fixed point on the string. Let us divide the
string coordinates on specially noted time and space components:
\begin{equation*}
  (X^\Sa):=(X^0:=T,X^\Si:=Y^\Si),\qquad \Si=1,\dotsc,D-1.
\end{equation*}
From the point of view of an external observer which is located in the target
space this point moves with the observed velocity
\begin{equation}                                                  \label{unbcgt}
  v^\Si:=c\frac{dY^\Si}{dX^0}=c\frac{\dot Y^\Si}{\dot X^0},
\end{equation}
where $c$ is the light velocity, and differentiation is performed along the
world line of the point on the string worldsheet. It is plausible to define the
nonrelativistic limit in the same way as for a point particle:
\begin{equation}                                                  \label{esnbcg}
  \frac{v^\Si}c\to0, \qquad\forall\Si.
\end{equation}
However this limit does not satisfy the requirement of reparameterization
invariance of the string worldsheet because the limit (\ref{unbcgt}) contains
differentiation only with respect to $\tau$. The solution of this problem
results in the definition of the nonrelativistic limit in terms of areas but not
lengthes as in Eq.~(\ref{esnbcg}).

The requirement of the invariance of the action with respect to general
coordinate transformations on the string worldsheet is geometrical. Then the
action describes the string worldsheet but not the coordinate system chosen on
it.

We consider sufficiently smooth embedding of the string worldsheet
$\overline\MU$ in $D$-dimensional Minkowski space $\MR^{1,D-1}$:
\begin{equation}                                                  \label{ubnchs}
  X:\qquad\MR^2\supset\overline\MU\ni\qquad(\tau,\s)\mapsto
  \big(X^\Sa(\tau,\s)\big)\qquad\in\MR^{1,D-1},
\end{equation}
where $X^\Sa$, $\Sa=0,1,\dotsc,D-1$, are Cartesian coordinates in Minkowski
space with metric $\eta_{\Sa\Sb}:=\diag(+-\dotsc-)$ and $\tau,\s$ are
coordinates on the string worldsheet. We assume that the coordinate $\tau\in\MR$
is timelike and $\s$ is spacelike. That is
$\dot X^2:=\dot X^\Sa\dot X^\Sb\eta_{\Sa\Sb}>0$ and
$X^{\prime2}:=X^{\prime\Sa}X^{\prime\Sb}\eta_{\Sa\Sb}<0$, where the dot and
prime denote respectively differentiation on $\tau$ and $\s$. By assumption,
the space coordinate for open and closed strings is varied in the intervals
$\s\in[0,\pi]$ and $\s\in[-\pi,\pi]$, respectively.

The Nambu--Goto action is proportional to the area of a string worldsheet
\begin{equation}                                                  \label{ubxvgy}
  S_{\Sn\Sg}:=-\rho c\int_{\overline\MU}\!\!dx\sqrt{|h|}
  =-\rho c\int_{\overline\MU}\!\!d\tau d\s\sqrt{\displaystyle(\dot X,X')^2
  -\dot X^2X^{\prime2}},
\end{equation}
where $\rho=\const$ is the linear mass density of a string, $c$ is the velocity
of light and parenthesis denote the scalar product in Minkowskian space.

Now we define the expansion parameter. The projection of the area element of a
string worldsheet $d\upsilon:=d\tau d\s\sqrt{|h|}$ on spacelike hypersurface
$T=\const$ has the form
\begin{equation*}
  d\upsilon_\perp=d\tau d\s\sqrt{\big|(\dot YY')^2-\dot Y^2Y^{\prime\,2}}\big|
  =d\tau d\s\sqrt{\dot Y_\perp^2Y^{\prime\,2}},
\end{equation*}
where
\begin{equation*}
  \dot Y_\perp^\Si:=\dot Y^\Si-\frac{(\dot Y,Y')Y^{\prime\Si}}{Y^{\prime2}}
\end{equation*}
is the orthogonal component of the velocity vector and, for brevity, we drop
indices $\Si$ enumerating space coordinates of the string. To derive this
formula, it is sufficient to put $T=\const$ in the determinant of the induced
metric $h$. It is important to note that the volume element $d\upsilon_\perp$
has the correct transformation properties under coordinate changes because
we have the determinant of the metric induced by the embedding of the string
worldsheet in the Euclidean space $\MR^{D-1}\subset\MR^{1,D-1}$ under the square
root. Let us denote the projection of $d\upsilon$ on the coordinate plane
$(T,Y^\Si)$ by $d\upsilon_\Si$. Then
\begin{equation*}
  d\upsilon_\Si:=d\tau d\s \sqrt{\dot T^2Y^{\prime\Si\,2}+
  T^{\prime\,2}\dot Y^{\Si\,2}-2\dot T T'\dot Y^\Si Y^{\prime\Si}},
\end{equation*}
where summation over $\Si$ is absent.

We introduce the ratio
\begin{equation}                                                  \label{esfwrn}
  \e:=\frac{(d\upsilon_\perp)^2}{(d\upsilon_0)^2}
  =\frac{\dot Y_\perp^2Y^{\prime\,2}}{A^2}>0,
\end{equation}
where
\begin{align}                                                          \nonumber
  (d\upsilon_0)^2:=&\sum_{\Si=1}^{D-1}(d\upsilon_\Si)^2:=d\s^2d\tau^2A^2,
\\                                                                \label{enndhy}
  A^2:=&(\dot TY'-T'\dot Y)^2
  =\dot T^2Y^{\prime2}+T^{\prime2}\dot Y^2-2\dot TT'(\dot Y,Y').
\end{align}
Here and in what follows, the summation over space components of the string is
performed using the Euclidean metric:
\begin{equation*}
  \dot Y^2:=\dot Y^\Si\dot Y^\Sj\dl_{\Si\Sj},\qquad
  Y^{\prime2}:=Y^{\prime\Si}Y^{\prime\Sj}\dl_{\Si\Sj},\qquad
  (\dot Y,Y'):=\dot Y^\Si Y^{\prime\Sj}\dl_{\Si\Sj}.
\end{equation*}
Since area elements $d\upsilon_\perp$ and $d\upsilon_0$ have the same
transformation rules, the ratio $\e(x)$ is a scalar field (function).
Therefore it can be used as the invariant expansion parameter assuming
$\e\ll1$.

Let us rewrite the Nambu--Goto action (\ref{ubxvgy}) in new notation:
\begin{equation*}
  S_{\Sn\Sg}=-\rho c\int\!d\tau d\s\sqrt{A^2-\dot Y_\perp^2Y^{\prime\,2}}.
\end{equation*}
Then we obtain the Lagrangian for a nonrelativistic string in the first order in
expansion in $\e$:
\begin{equation}                                                  \label{ubcnde}
  L_{\Sn\Ss}=\rho c\sqrt{A^2}\left(\frac{\dot Y_\perp^2Y^{\prime\,2}}{2A^2}
  -1\right).
\end{equation}
This Lagrangian is the answer to the question raised beforehand. By
construction, it is invariant with respect to general coordinate transformations
of $(\tau,\s)$ on the string worldsheet. Moreover, the Lagrangian (\ref{ubcnde})
is invariant under translations $X^\Sa\to X^\Sa+\const^\Sa$ and global
$\MS\MO(D-1)$-rotations acting on the space string coordinates $Y$.

We now consider the Lagrangian for a nonrelativistic string (\ref{ubcnde}) by
itself do not paying attention on how it was derived and do not assuming
smallness of the parameter $\e$. There are open, $\s\in[0,\pi]$, and closed,
$\s\in[-\pi,\pi]$, strings as in the relativistic case.

Lagrangian (\ref{ubcnde}) takes simple and visual form in the time gauge
$\tau=T$, $(\dot Y,Y')=0$:
\begin{equation}                                                  \label{ubxhsp}
  L_{\Sn\Ss}\big|_\text{time gauge}
  =\rho c\sqrt{Y^{\prime\,2}}\left(\frac12\dot Y^2-1\right),
\end{equation}
where the first summand is the kinetic term for the transverse oscillations of
the string, and the second one is equal to the potential energy which is
proportional to the length of the string. The common factor $\sqrt{Y^{\prime2}}$
is due to the arbitrariness in the choice of the space coordinate $\s$. In this
gauge, the expansion parameter (\ref{esfwrn}) takes the form
\begin{equation*}
  \e=\frac{v^2}{c^2},
\end{equation*}
that is the limit $\e\to0$ is really nonrelativistic.

From now on we put $c=1$ for simplicity.

To derive equations of motion for the nonrelativistic string, we rewrite
Lagrangian (\ref{ubcnde}) through independent variables $(T,Y)$, with respect to
which it is varied:
\begin{equation}                                                  \label{ubbcvf}
  L_{\Sn\Ss}=\rho\left[\frac{\dot Y^2Y^{\prime2}-(\dot Y,Y')^2}
  {2\sqrt{\dot T^2Y^{\prime2}+T^{\prime2}\dot Y^2-2\dot TT'(\dot Y,Y')}}
  -\sqrt{\dot T^2Y^{\prime2}+T^{\prime2}\dot Y^2-2\dot TT'(\dot Y,Y')}\right]
\end{equation}
We introduce notation to simplify the subsequent formulae:
\begin{align}                                                     \label{ubbcft}
  P^\tau_0:=&\frac{\pl L_{\Sn\Ss}}{\pl\dot T}
  =-\frac\rho{\sqrt{A^2}}
  \left(\frac{\dot Y_\perp^2Y^{\prime2}}{2A^2}+1\right)
  \big(\dot T Y^{\prime2}-T'(\dot Y,Y')\big),
\\                                                                \label{ubcvdg}
  P^\tau_\Si:=&\frac{\pl L_{\Sn\Ss}}{\pl\dot Y^\Si}
  =\frac\rho{\sqrt{A^2}}\left[Y^{\prime2}\dot Y_{\perp\Si}
  -\left(\frac{\dot Y^2_\perp Y^{\prime2}}{2A^2}+1\right)
  \big(T^{\prime2}\dot Y_\Si-\dot TT'Y^\prime_\Si\big)\right],
\\
  P^\s_0:=&\frac{\pl L_{\Sn\Ss}}{\pl T'}
  =-\frac\rho{\sqrt{A^2}}
  \left(\frac{\dot Y_\perp^2Y^{\prime2}}{2A^2}+1\right)
  \big(T'\dot Y^2-\dot T(\dot Y,Y')\big),
\\                                                                \label{qbsewn}
  P^\s_\Si:=&\frac{\pl L_{\Sn\Ss}}{\pl Y^{\prime\Si}}
  =\frac\rho{\sqrt{A^2}}
  \left[\dot Y^2Y^\prime_\Si-(\dot Y,Y')\dot Y_\Si
  -\left(\frac{\dot Y_\perp^2Y^{\prime2}}{2A^2}+1\right)
  \big(\dot T^2Y^\prime_\Si-\dot TT'\dot Y_\Si\big)\right].
\end{align}
Then the Euler--Lagrange equations are rewritten as
\begin{align}                                                     \label{unmvjj}
  \frac{\dl S_{\Sn\Ss}}{\dl T}=&-\frac\pl{\pl\tau}P^\tau_0
  -\frac\pl{\pl\s}P^\s_0=0,
\\                                                                \label{ubbdng}
  \frac{\dl S_{\Sn\Ss}}{\dl Y^\Si}=&-\frac\pl{\pl\tau}P^\tau_\Si
  -\frac\pl{\pl\s}P^\s_\Si=0.
\end{align}

As in the case of the relativistic string, we assume vanishing variations on
the boundaries $\tau=\tau_{1,2}$, and consider arbitrary variations on the
boundaries $\s=0,\pi$ (free ends). Then the variational principle implies the
boundary conditions
\begin{equation}                                                  \label{ubndhy}
  P^\s_0\big|_{\s=0,\pi}=0,\qquad P^\s_\Si\big|_{\s=0,\pi}=0,\qquad\forall\Si
\end{equation}
for an open string.

In the time gauge $\tau=T$, $(\dot Y,Y')=0$ the following equalities hold
\begin{equation*}
  P^\s_0\equiv0,\qquad P^\s_\Si=-\frac\rho{\sqrt{Y^{\prime2}}}
  \left(\frac{\dot Y^2}2-1\right)Y'_\Si.
\end{equation*}
Therefore, the boundary conditions take the form
\begin{equation*}
  \left(\frac{\dot Y^2}2-1\right)\left.\frac{Y'_\Si}
  {\sqrt{Y^{\prime2}}}\right|_{\s=0,\pi}=0.
\end{equation*}
These imply two possibilities for space string components. The first is
\begin{equation}                                                  \label{ubvftr}
  \dot Y^2\big|_{\s=0,\pi}=2.
\end{equation}
The second implies
\begin{equation*}
  \left.\frac{Y'_\Si}{\sqrt{Y^{\prime2}}}\right|_{\s=0,\pi}=0,
  \qquad\forall\Si.
\end{equation*}
Since the vector in the left hand side of the last equality has the unit length,
the continuity condition is broken in the second case, and therefore we do not
consider it. Thus, the boundary conditions for the nonrelativistic string
in the time gauge has the form (\ref{ubvftr}). These boundary conditions for a
free nonrelativistic string follow from the least action principle and mean that
the string end points move with the velocity $\sqrt2c$ which is greater then the
velocity of light. There is no contradiction because Newton's mechanics admits
arbitrary velocities.

Equations of motion for the nonrelativistic bosonic string (\ref{unmvjj}) and
(\ref{ubbdng}) are not independent. Invariance of the action under the choice of
coordinates on the string worldsheet, due to Noether's theorem, implies two
linear identities:
\begin{equation}                                                  \label{ummsjt}
  \frac{\dl S_{\Sn\Ss}}{\dl T}\pl_\al T
  +\frac{\dl S_{\Sn\Ss}}{\dl Y^\Si}\pl_\al Y^\Si\equiv0,\qquad\al=0,1.
\end{equation}

Since Lagrangian (\ref{ubcnde}) and the corresponding action are invariant with
respect to global translations in the target space
\begin{equation*}
  T\mapsto T+\e^0,\qquad Y^\Si\mapsto Y^\Si+\e^\Si,\qquad x^\al\mapsto x^\al,
\end{equation*}
where $(\e^0,\e^\Si)$ are constant parameters, the first Noether theorem implies
the conservation of the currents
\begin{equation}                                                  \label{ebcvdg}
  \pl_\al J_\Sa{}^\al=0,\qquad\forall\Sa,
\end{equation}
where
\begin{align*}
  J_0{}^\tau=&-\frac{\pl L_{\Sn\Ss}}{\pl\dot T}=-P_0^\tau, &
  J_0{}^\s=&-\frac{\pl L_{\Sn\Ss}}{\pl T'}=-P_0^\s,
\\
  J_\Si{}^\tau=&-\frac{\pl L_{\Sn\Ss}}{\pl\dot Y^\Si}=-P_\Si^\tau, &
  J_\Si{}^\s=&-\frac{\pl L_{\Sn\Ss}}{\pl Y^{\prime\Si}}=-P_\Si^\s.
\end{align*}
Formally, the Euler--Lagrange equations (\ref{unmvjj}), (\ref{ubbdng}) coincide
with the current (\ref{ebcvdg}) conservation. This yields the physical
interpretation to the introduced notation: up to a sign, $P_0^\tau$ is the
linear energy density and $P_\Si^\tau$ in the linear momentum density of the
nonrelativistic string. The total energy and momentum are given by the
integrals
\begin{align}                                                     \label{unncbd}
  \CE=:-\int\!d\s\frac{\pl L_{\Sn\Ss}}{\pl\dot T},
\\                                                                \label{ubsvwf}
  \CP_\Si=:-\int\!d\s\frac{\pl L_{\Sn\Ss}}{\pl\dot Y^\Si},
\end{align}
where integration is performed from zero to $\pi$ or from $-\pi$ to $\pi$ for
open and closed strings, respectively.

Moreover, the action for the nonrelativistic string is invariant under global
$\MS\MO(D-1)$-rotations which take the following infinitesimal form
\begin{equation*}
  T\mapsto T,\qquad Y^\Si\mapsto Y^\Si-Y^\Sj\om_\Sj{}^\Si,\qquad
  x^\al\mapsto x^\al,
\end{equation*}
where $\om^{\Si\Sj}=-\om^{\Sj\Si}$ are rotational parameters. Due to the first
Noether's theorem, the invariance of the action results in the current
conservation on the equations of motion
\begin{equation}                                                  \label{ubbvsg}
  \pl_\al J_{\Si\Sj}{}^\al=0,
\end{equation}
where
\begin{equation}                                                  \label{unnchg}
\begin{split}
  J_{\Si\Sj}{}^\tau=&Y_\Si\frac{\pl L_{\Sn\Ss}}{\pl\dot Y^\Sj}
  -Y_\Sj\frac{\pl L_{\Sn\Ss}}{\pl\dot Y^\Si},
\\
  J_{\Si\Sj}{}^\s=&Y_\Si\frac{\pl L_{\Sn\Ss}}{\pl Y^{\prime\Sj}}
  -Y_\Sj\frac{\pl L_{\Sn\Ss}}{\pl Y^{\prime\Si}}.
\end{split}
\end{equation}
It implies conservation of the total angular momentum of the
nonrelativistic string
\begin{equation}                                                  \label{ubnchg}
  \CM_{\Si\Sj}:=\int_{0,-\pi}^\pi\!\!\!d\s\left(
  Y_\Si\frac{\pl L_{\Sn\Ss}}{\pl\dot Y^\Sj}
  -Y_\Sj\frac{\pl L_{\Sn\Ss}}{\pl\dot Y^\Si}\right).
\end{equation}
Note that $J_{0\Si}{}^\al\equiv0$ for space rotations.
\section{Canonical formulation}
Momenta conjugate to $T$ and $Y^\Si$ are
\begin{align}                                                     \label{uvvcfd}
  P:=&P_0^\tau=\frac{\pl L_{\Sn\Ss}}{\pl\dot T}
  =-\frac\rho{\sqrt{A^2}}\left(\frac{\dot Y_\perp^2Y^{\prime2}}{2A^2}
  +1\right)\big(\dot TY^{\prime2}-T'(\dot Y,Y')\big),
\\                                                                \label{unsbsw}
  P_\Si:=&P_\Si^\tau=\frac{\pl L_{\Sn\Ss}}{\pl\dot Y^\Si}
  =~\frac\rho{\sqrt{A^2}}\left[Y^{\prime2}\dot Y_{\perp\Si}
  -\left(\frac{\dot Y_\perp^2Y^{\prime2}}{2A^2}+1\right)
  (T^{\prime2}\dot Y_\Si-\dot TT'Y'_\Si)\right].
\end{align}
As in the case of relativistic string, the canonical Hamiltonian for the
nonrelativistic string is equal to zero
\begin{equation*}
  H:=P\dot T+P_\Si\dot Y^\Si-L_{\Sn\Ss}=0,
\end{equation*}
as the consequence of straightforward calculations. It means that the dynamics
of the model is entirely defined by the constraints which are present in the
theory. Since the action of the nonrelativistic string
\begin{equation*}
  S_{\Sn\Ss}:=\rho\int\!d\tau d\s
  \sqrt{A^2}\left(\frac{\dot Y_\perp^2Y^{\prime2}}{2A^2}-1\right)
\end{equation*}
is invariant with respect to arbitrary coordinate transformations $\tau,\s$ on
the string worldsheet, we expect that the model contains two primary first class
constraint. Unfortunately, the form of momenta (\ref{uvvcfd}), (\ref{unsbsw}) is
relatively complicated, and seeing these constraints is problematic. Therefore
we introduce notation for next calculations
\begin{equation}                                                  \label{unnbdf}
  \dot T_\perp:=\dot T-\frac{(\dot Y,Y')}{Y^{\prime2}}T',
\end{equation}
\begin{equation}                                                  \label{uajqwh}
\begin{split}
  A^2=&\dot T_\perp^2Y^{\prime2}+T^{\prime2}\dot Y_\perp^2,
\\
  \e A^2=&\dot Y_\perp^2Y^{\prime2},
\end{split}
\end{equation}
where the variable $\e$ (\ref{esfwrn}) is by this time not small. Then
expressions for momenta (\ref{uvvcfd}), (\ref{unsbsw}) take the form
\begin{equation}                                                  \label{unncbl}
\begin{split}
  P=&-\frac\rho{\sqrt{A^2}}\left(\frac\e2+1\right)\dot T_\perp Y^{\prime2},
\\
  P_\Si=&~~\frac\rho{\sqrt{A^2}}\left[Y^{\prime2}\dot Y_{\perp\Si}
  -\left(\frac\e2+1\right)\big(T^{\prime2}\dot Y_{\perp\Si}-\dot T_\perp T'
  Y'_\Si\big)\right].
\end{split}
\end{equation}

Now we derive expressions invariant with respect to $\MS\MO(D-1)$-rotations
which are quadratic in generalized coordinates and momenta
\begin{align}                                                     \label{ubvcfu}
  P^2=&~~\rho^2\left(\frac\e2+1\right)^2(Y^{\prime2}-\e T^{\prime2}),
\\                                                                \label{unbcgs}
  P_\Si^2:=P^\Si P_\Si=&~~\rho^2\left(T^{\prime2}+\e Y^{\prime2}
  -\e T^{\prime2}-\frac34\e^2T^{\prime2}\right),
\\                                                                \label{evdbfr}
  PT'=&-\frac\rho{\sqrt{A^2}}\left(\frac\e2+1\right)\dot T_\perp T'Y^{\prime2},
\\                                                                \label{eabdnf}
  P_\Si Y^{\prime\Si}=&
  ~~\frac\rho{\sqrt{A^2}}\left(\frac\e2+1\right)\dot T_\perp T'Y^{\prime2}.
\end{align}
The last two equalities imply the primary constraint
\begin{equation}                                                  \label{unncht}
  H_1:=PT'+P_\Si Y^{\prime\Si}=0.
\end{equation}
This constraint is kinematical and has the same form as for the relativistic
string. The dynamical constraint $H_0$ is more complicated. For its deriving,
one can find $\e$ from equality (\ref{unbcgs}) by solving quadratic equation.
Afterwards this solution is to be substituted into equality (\ref{ubvcfu}), and
this results in dynamical constraint in the form of the polynomial in canonical
variables and their space derivatives. The resulting expressions are cumbersome,
and we leave them for future investigations. General considerations imply that
constraints $H_0$ and $H_1$ must be of the first class, and their Poisson
bracket algebra should be isomorphic to conformal algebra.
\section{The time gauge I}
Equations of motion for the nonrelativistic string (\ref{unmvjj}),
(\ref{ubbdng}) are complicated, and, for their simplification, we fix the time
gauge using the
invariance of the model with respect general coordinate transformations on the
string worldsheet. First, we fix the conformal gauge for the induced metric
\begin{equation}                                                  \label{unncbc}
\begin{split}
  \dot X^2+X^{\prime2}=0\qquad\Leftrightarrow&\qquad
  \dot T^2-\dot Y^2+T^{\prime2}-Y^{\prime2}=0,
\\
  (\dot X,X')=0\qquad\Leftrightarrow&\qquad \dot TT'-(\dot Y,Y')=0.
\end{split}
\end{equation}
Next, we impose the additional condition using the residual conformal invariance
\begin{equation}                                                  \label{eggdft}
  \tau=T
\end{equation}
in the same way as was done for the relativistic string. Then the conformal
gauge (\ref{unncbc}) becomes
\begin{align}                                                     \label{uncbdo}
  \dot Y^2+Y^{\prime2}=&1,
\\                                                                \label{unbsht}
  (\dot Y,Y')=&0.
\end{align}
We call conditions (\ref{eggdft})--(\ref{unbsht}) for the nonrelativistic string
{\em the time gauge I}. In Section \ref{swbstg}, we used partial time gauge
without condition (\ref{uncbdo}).

In what follows, we need consequences of conditions (\ref{uncbdo}),
(\ref{unbsht}) obtained by differentiation:
\begin{align}                                                     \label{ubvfrj}
  \pl_0(\dot Y^2+Y^{\prime2}=1):&\qquad(\dot Y,\ddot Y)+(Y',\dot Y')=0,
\\                                                                \label{ubanww}
  \pl_1(\dot Y^2+Y^{\prime2}=1):&\qquad(\dot Y,\dot Y')+(Y',Y'')=0,
\\                                                                \label{esswrt}
  \pl_0\big((\dot Y,Y')=0\big):&\qquad(\ddot Y,Y')+(\dot Y,\dot Y')=0,
\\                                                                \label{eahysr}
  \pl_1\big((\dot Y,Y')=0\big):&\qquad(\dot Y',Y')+(\dot Y,Y'')=0.
\end{align}
Formulae (\ref{ubvfrj}), (\ref{eahysr}) and (\ref{ubanww}), (\ref{esswrt})
imply equalities:
\begin{align}                                                     \label{unbchg}
  (\dot Y,\ddot Y)-(\dot Y,Y'')=&0,
\\                                                                \label{unncok}
  (Y',Y'')-(Y',\ddot Y)=&0.
\end{align}

In the time gauge, velocities of string points are perpendicular to the string:
\begin{equation*}
  \dot Y_\perp^\Si=\dot Y^\Si
\end{equation*}
and subsidiary fields combinations (\ref{enndhy}) and (\ref{esfwrn}) are
\begin{equation}                                                  \label{unbgtr}
  A^2=Y^{\prime2},\qquad\e=\dot Y^2.
\end{equation}

Now expressions (\ref{ubbcft})--(\ref{qbsewn}) take the form
\begin{align}                                                     \label{ukkiuy}
  P_0^\tau=&-\rho\sqrt{Y^{\prime2}}\left(\frac{\dot Y^2}2+1\right),
\\                                                                \label{unksuu}
  P_\Si^\tau=&~\rho\sqrt{Y^{\prime2}}\dot Y_\Si,
\\                                                                \label{edjsjh}
  P_0^\s=&~0,
\\                                                                \label{ervbxf}
  P_\Si^\s=&~\frac\rho{\sqrt{Y^{\prime2}}}\left(\frac{\dot Y^2}2-1\right)Y'_\Si
\end{align}
and equations of motion (\ref{unmvjj}), (\ref{ubbdng}) are essentially
simplified
\begin{align}                                                     \label{ubbcng}
  &Y^{\prime^2}(\dot Y,\ddot Y)+(Y',\dot Y')\left(\frac{\dot Y^2}2+1\right)=0,
\\                                                                \label{unncbp}
  &Y^{\prime2}\ddot Y_\Si+(Y',\dot Y')\dot Y_\Si+\left(\frac{\dot Y^2}2-1\right)
  Y''_\Si+\frac{\dot Y^2}{2Y^{\prime2}}(Y',Y'')Y'_\Si=0,
\end{align}
where formulae (\ref{ubanww}) and (\ref{uncbdo}) are used.

Note that equalities (\ref{ukkiuy})--(\ref{ervbxf}) hold for weaker conditions:
it is sufficient to put $\tau=T$ and $(\dot Y,Y')=0$.

\begin{prop}
Equation (\ref{ubbcng}) is the consequence of equation (\ref{unncbp}).
\end{prop}
\begin{proof}
Contract Eq.~(\ref{unncbp}) with $\dot Y^\Si$ and use Eqs.~(\ref{eahysr})
and (\ref{ubanww}).
\end{proof}

It is easily verified that contraction of Eq.~(\ref{unncbp}) with
$Y^{\prime\Si}$ reduces to condition (\ref{uncbdo}).

Equations of motion (\ref{unncbp}) can be rewritten as
\begin{equation}                                                  \label{unncbk}
  \ddot Y_\Si-Y''_\Si-\dot Y^2\ddot Y_\Si+\frac{\dot Y^2}2Y''_\Si
  +(Y',\dot Y')\dot Y_\Si+\frac{\dot Y^2}{2Y^{\prime2}}(Y',Y'')Y'_\Si=0,
\end{equation}
where equalities (\ref{uncbdo}) and (\ref{ubanww}) are taken into account. Thus,
equations of motion for the nonrelativistic string reduce to the system of
nonlinear equations (\ref{unncbk}) and quadratic constraints (\ref{uncbdo}),
(\ref{unbsht}). We see that equations of motion of the nonrelativistic string
in the time gauge are more complicated than equations for the relativistic
string because of the nonlinearity.

Let us present the class of solutions which is singled out by the extra
condition $\dot Y^2=\const$ linearizing Eqs.~(\ref{unncbk}) in the time gauge.
Since velocities are restricted by condition (\ref{uncbdo}), this class of
solutions cannot describe an open string of finite length because the boundary
condition (\ref{ubvftr}) cannot be satisfied. Therefore we consider an open
infinite string.

Let
\begin{equation}                                                  \label{ubcndg}
  \dot Y^2=\sin^2\g,\qquad Y^{\prime2}=\cos^2\g,\qquad\g=\const\in(0,\pi/2).
\end{equation}
It is clear that constraint (\ref{uncbdo}) holds for all $\g$. Differentiation
of these equalities yields the following relations:
\begin{equation*}
  (\dot Y,\dot Y')=0,\qquad (Y',Y'')=0.
\end{equation*}
This implies that Eqs.~(\ref{unncbk}) reduce to the free wave equation in this
case
\begin{equation}                                                  \label{unnvhg}
  v^2\ddot Y_\Si-Y''_\Si=0,\qquad 0<v^2:=\frac{2\cos^2\g}{1+\cos^2\g}<1.
\end{equation}
A general solution of this equation is
\begin{equation}                                                  \label{enbchd}
  Y_\Si=F_\Si(\xi)+G_\Si(\eta),
\end{equation}
where $F_\Si(\xi)$ and $G_\Si(\eta)$ are arbitrary functins of cone coordinates:
\begin{equation*}
  \xi:=v\tau+\s,\qquad\eta:=v\tau-\s.
\end{equation*}
Extra conditions (\ref{ubcndg}) and (\ref{unbsht}) impose the following
restrictions on arbitrary functions:
\begin{equation}                                                  \label{uncmdj}
\begin{split}
  F^{\prime2}+2(F',G')+G^{\prime2}=&\frac{\sin^2\g}{v^2},
\\
  F^{\prime2}-2(F',G')+G^{\prime2}=&\cos^2\g,
\\
  F^{\prime2}-G^{\prime2}=&0,
\end{split}
\end{equation}
where primes denote differentiation on respective arguments. The obtained
system of equations has the solution
\begin{equation}                                                  \label{unbchj}
\begin{split}
  F^{\prime2}=G^{\prime2}=\frac{\sin^2\g+v^2\cos^2\g}{4v^2}&=:a^2>0,
\\
  (F',G')=\frac{\sin^2\g-v^2\cos^2\g}{4v^2}&.
\end{split}
\end{equation}
Sure, velocity $v$ in the right hand side can be expressed by angle $\g$ using
the definition (\ref{unnvhg}) but it does not simplify the following formulae.

Thus, we get the class of solutions for the nonrelativistic string in the time
gauge I, which has the form (\ref{enbchd}) where arbitrary functions are
restricted by Eqs.~(\ref{unbchj}) parameterized by the angle $\g\in(0,\pi/2)$.
These conditions have nontrivial solutions. For example, there is the solution
in four dimensions
\begin{equation}                                                  \label{unbcgf}
\begin{split}
  F=&(a\xi,~0,~0),
\\
  G=&(b\eta,~c\cos\eta,~c\sin\eta),
\end{split}
\end{equation}
where constants $a$, $b$ and $c$ are defined by equalities
\begin{equation*}
\begin{split}
  a:=&\frac{\sqrt{\sin^2\g+v^2\cos^2\g}}{2v},
\\
  b:=&\frac{\sin^2\g-v^2\cos^2\g}{2v\sqrt{\sin^2\g+v^2\cos^2\g}},
\\
  c:=&\frac{\sin\g\cos\g}{\sqrt{\sin^2\g+v^2\cos^2\g}},
\end{split}
\end{equation*}
which can be simply verified. Thus, space coordinates of the nonrelativistic
string in this case are
\begin{equation}                                                  \label{ubhyuu}
  Y=(a\xi+b\eta,~c\cos\eta,~c\sin\eta).
\end{equation}
In this way, we derive the class of exact solutions for the nonrelativistic
string in the time gauge I which is parameterized by constant $\g\in(0,\pi/2)$.
The initial configuration for $\tau=0$ is
\begin{equation}                                                  \label{unnchj}
  Y(0,\s)=\big((a-b)\s,~c\cos\s,~-c\sin\s),
\end{equation}
where
\begin{equation*}
  a-b=\frac{v\cos^2\g}{\sqrt{\sin^2\g+v^2\cos^2\g}}=\frac{v\cos\g}{\sin\g}c.
\end{equation*}
This is the spiral depicted in Fig.~\ref{fhelix}
\begin{figure}[hbt]
\hfill\includegraphics[width=.4\textwidth]{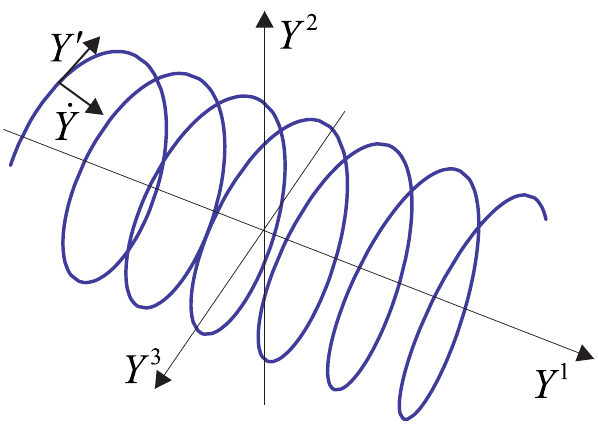}
\hfill {}
\centering\caption{Nonrelativistic bosonic string in the form of infinite
spiral.}
\label{fhelix}
\end{figure}
During the evolution, the spiral moves translational with constant velocity
along $Y^1$ axis and rotates simultaneously around the same axis. If $v\to1$,
(relativistic limit) then $\g\to0$ and the radius of the spiral goes to zero,
i.e.\ the spiral degenerates into the line.
\section{The time gauge II}
The theory of the nonrelativistic string allows one to impose the conformal
gauge using the reparameterization invariance in the same way like it is done
for the relativistic string. This gauge was used in the previous section. The
conformal gauge (\ref{unncbc}) is invariant with respect to the global action
of the Poincar\'e group in the target space. However, since we consider now the
nonrelativistic string, these conditions can be changed. Therefore we impose
the gauge
\begin{equation}                                                  \label{ubbcvl}
  Y^{\prime2}=1,\qquad (\dot Y,Y')=0
\end{equation}
instead of the conformal gauge (\ref{unncbc}). The first condition means that
the length of the string (from the point of view of external observer)
is chosen as parameter $\s$. The second condition means that the velocity vector
is perpendicular to the string. In contrast to the time gauge (\ref{unncbc})
now we do not have restrictions on the square of velocity vector $\dot Y^2$.

In the gauge (\ref{ubbcvl}), the metric on the string worldsheet induced by the
embedding (\ref{ubnchs}) takes the form
\begin{equation}                                                  \label{ubvxfr}
  h_{\al\bt}=\begin{pmatrix}
    \dot T^2-\dot Y^2 & \dot TT' \\ \dot TT' & T^{\prime2}-1 \end{pmatrix}.
\end{equation}
If we consider the metric induced by the embedding of the string worldsheet
in the Euclidean subspace $\MR^{D-1}\subset\MR^{1,D-1}$ (with positive definite
metric) then the Riemannian metric is obtained
\begin{equation}                                                  \label{ubbcny}
  \hat h_{\al\bt}=\begin{pmatrix} \dot Y^2 & 0 \\ 0 & 1 \end{pmatrix}.
\end{equation}
It is known that such coordinate system exists locally. In addition, it is not
defined uniquely. Suppose that the residual symmetry allows us to impose one
more condition
\begin{equation}                                                  \label{ubncht}
  \tau=T.
\end{equation}
We call conditions (\ref{ubbcvl}) and (\ref{ubncht}) {\em the time gauge II}
for the nonrelativistic string. In contrast to the relativistic string, these
conditions are invariant only under global translations and
$\MS\MO(D-1)$-rotations in the target space.

Let us differentiate conditions (\ref{ubbcvl}) on $\tau$ and $\s$, respectively:
\begin{equation*}
  2(Y',\dot Y')=0,\qquad (\dot Y',Y')+(\dot Y,Y'')=0.
\end{equation*}
It implies the equality
\begin{equation}                                                  \label{unbcfr}
  (\dot Y,Y'')=0,
\end{equation}
which holds in the time gauge II.

Equalities (\ref{ukkiuy})--(\ref{ervbxf}) are fulfilled both in the time gauge
I and II.

Now equations of motion (\ref{unmvjj}), (\ref{ubbdng}) are
\begin{align}                                                     \label{unncbg}
  \pl_0\dot Y^2=2(\dot Y,\ddot Y)=0,
\\                                                                \label{uvvxbg}
  \ddot Y_\Si-Y^{\prime\prime}_\Si+\frac12\pl_1(\dot Y^2 Y'_\Si)=0.
\end{align}
\begin{prop}
Equation (\ref{unncbg}) is the consequence of Eqs.~(\ref{uvvxbg}).
\end{prop}
\begin{proof}
Substitute $\ddot Y$ from Eq.~(\ref{uvvxbg}) into Eq.~(\ref{unncbg})
\begin{equation*}
  (\dot Y,\ddot Y)=(\dot Y,Y'')-(\dot Y,Y'')\frac{\dot Y^2}2
  -(\dot Y,Y')(\dot Y,\dot Y')=0,
\end{equation*}
as the consequence of Eq.~(\ref{ubbcvl}) and (\ref{unbcfr}).
\end{proof}

In the linear approximation, Equations (\ref{uvvxbg}) describe transverse
(due to extra conditions (\ref{ubbcvl})) oscillations of the string which
propagate along the string with the velocity of light $c$.

Thus the nonrelativistic string in the time gauge II is described only by
the spatial components $Y^\Si(\tau,\s)$ which satisfy equations of motion
(\ref{uvvxbg}) and additional conditions (\ref{ubbcvl}). For open and closed
strings, equations of motion must be supplemented by boundary conditions
(\ref{ubvftr}) and the periodic conditions, respectively.

We give an example of exact solution.
Consider a straight open string of length $L$ which rotates with constant
angular speed in the three dimensional Minkowskian space $\MR^{1,2}$.
Suppose that rotation takes place in the $(Y^1,Y^2)$ plane with constant
angular velocity $\om$
\begin{equation*}
  Y=(\s\cos\om\tau,~\s\sin\om\tau),\qquad\s\in[-L/2,L/2].
\end{equation*}
Here we changed the interval for the space coordinate to make rotation happen
around the center of mass. Then the following equalities hold
\begin{align*}
  \dot Y=&(-\s\om\sin\om\tau,~\s\om\cos\om\tau)
\\
  Y'=&(\cos\om\tau,~\sin\om\tau),
\\
  \dot Y^2=&\s^2\om^2\sin^2\om\tau+\s^2\om^2\cos^2\om\tau=\s^2\om^2,
\\
  Y^{\prime2}=&\cos^2\om\tau+\sin^2\om\tau=1,
\\
  (\dot Y,Y')=&-\s\om\cos\om\tau\sin\om\tau+\s\om\sin\om\tau\cos\om\tau=0,
\\
  \ddot Y=&(-\s\om^2\cos\om\tau,~-\s\om^2\sin\om\tau),
\\
  Y''=&(0,~0).
\end{align*}
The constraints (\ref{ubbcvl}) are fulfilled. It is easily verified that
equations of motion (\ref{uvvxbg}) are also satisfied. Boundary conditions
(\ref{ubvftr}) define the rotational speed: $L\om=2\sqrt2c$. We see that
the angular velocity of rotations is inversely proportional to the length
of the nonrelativistic string as that for the relativistic string.
\section{Conclusion}
In the paper, we define the nonrelativistic limit for the bosonic
Nambu--Goto string which results into the reparameterization invariant
Lagrangian (\ref{ubcnde}) describing nonrelativistic bosonic string. We present
the canonical formulation of the corresponding model and obtain particular
solutions of the Euler--Lagrange equations: one class of solutions of the
infinite moving and rotating spiral and the other solution in the form of finite
straight rod rotating with constant angular velocity.

The obtained Lagrangian can be considered by itself as the model of the string.
If this model is applied in nonrelativistic physics, e.g.\ in polymer physics,
the light velocity should be replace by the velocity of speed.

\end{document}